\documentclass {article}

\AtBeginDocument{%
  \providecommand\BibTeX{{%
    \normalfont B\kern-0.5em{\scshape i\kern-0.25em b}\kern-0.8em\TeX}}}

\RequirePackage{fix-cm}
\usepackage{amsthm}
\newtheorem{theorem}{Theorem}[section]

\newtheorem{lemma}[theorem]{Lemma}
\theoremstyle{definition}
\newtheorem{definition}{Definition}[section]

\usepackage{graphicx}
\usepackage{amssymb}
\usepackage{tikz-cd}

\usepackage{listings}

\begin{document}

\title{Re-looking at the  View Update Problem}

\author{Terry Brennan \\ {terryjbrennan@gmail.com} }

\begin{abstract}
Relational databases have always had a means for creating a pseudo-table, called a view, defined by a query.  Views are like tables in most ways, except that they are read-only and can’t be updated.  The problem of how to update views has attracted a lot of attention in the 1980s but is unsolved.

The best approach from that time was by Bancilhon and Spyratos.  I use one of their overlooked theorems and find a number of simple  solutions for common relational operators.
\end{abstract}





%
\maketitle

\section{The View Update Problem}
\label{intro}
Relational databases have always had operations and languages for combining and presenting data.  E.F. Codd defined a logical language similar to the first order predicate calculus, and showed how to convert any expression in that calculus to a set of operations on tables.  ANSI then defined a standard query language, SQL, that implemented and extended Codd’s relational algebra.  Part of the ANSI standards has been a view, an SQL query that becomes part of the database.  Views offer a different way of seeing the data.

From the time the idea of views arose, programmers have been trying to update databases through views.  The driving idea was data independence, that a programmer shouldn’t have to know the internal structure of the data.  Views offer a way of hiding the details and “smoothing” the data.  If programmers only updated database through view, then the internal structure of the database could change, but the programs wouldn’t need to be changed.  A simple example is a view that joins two data tables; the view hides the fact that there are two tables, presenting a single table to a programmer.  But data independence through views requires that views be fully updateable.

Almost immediately the hope for updating views ran into problem with their complexity.  Many approaches were tried but none were successful.  The idea of view updates slowly disappeared.

I propose reviving the view update problem for a practical reason – creating forms.  The SQL standard has been wildly successful for reporting.  A number of widely used tools are based on SQL, including presenting visual interfaces that create SQL statements.  Many non-programmers have been trained to create SQL. 

The SQL standard has not been as successful at automating and simplifying the creation of forms, where data is entered, changed or deleted, rather than reported.  Creating forms has required talented programmers working intensely, with a deep knowledge of the structure of the database and of its integrity constraints.  Some tools tried to automate form creation, but they weren’t generally successful because they lacked a programmer’s deep knowledge.  Updateable views would make form building easier, and possibly lead to automation of the process.

\section{Prior Work}
Dayal and Bernstein \cite{1} identified view updating as a problem.

A early result was by Bancilhon and Spyratos \cite{5}.  They started with a simple intuition – that a view update shouldn’t change information not in the view.  Views show only part of the information in a database, and the effect of view updates in the base database should be limited to the information exposed in the view.  Anything else is an unintended side effect.  They named this unchangeable information the view’s complement, and their rule was that the complement stayed constant;their approach became known as the constant complement approach.

Their development was entirely set-theoric.  They start with a set of states S, whose  internal structure was not specified.  In other words, the traditional structure of a database state -- as a tuple of relations over a set of domains -- wasn't assumed.  For them, an update is simply a function u from S to S, and a view simply a function f from S to another set V.

They then started building well-known set-theoretic structures.  A view, as a function between sets, induces an equivalence relation in S, denoted S/f, with all the states s in a class A $\in$ S/f map to the same view state.

Views can be ordered .  Let f,g be views, that is, f,g:S $\rightarrow$ V.  Let s, s' $\in$ S.  f$\geq$g iff f(s)=f(s') implies g(s)=g(s').  In each view's induced equivalence relation, f$\geq$g means f is a refinement of g, that each class in g's equivalence relation is the union of one or more classes in f's  equivalence relation.

This ordering creates a lattice; the largest view is called 1 and happens when each state in S in in their own equivalence class, meaning that it maps one-to-one to V.  The smallest view is called 0 and happens when all states in S map to the same constant state in V.  When f$\leq$g and g$\leq$f, they are called equivalent.  All of this is standard set theory.

As an example, consider a very small database DB with only one table, that has only one column, which can only take the value "a" and "b".   DB has four states: $\{$(a), (b)$\}$, $\{$(a)$\}$, $\{$(b)$\}$, $\{$$\}$.  For simplicity, I'll call the states ab, a, b and $\emptyset$.

Consider a view f that selects rows with the value a.  The view database will have the same structure as DB, but will have only two states: $\{$(a)$\}$,  $\{$$\}$.  The following table shows how f maps states of S to states of V:

\begin{center}
\begin{tabular}{ |c|c| } 
\hline
Base State & View State\\ 
\hline
\hline
ab & a\\
a & a\\
b & $\emptyset$ \\
$\emptyset$ & $\emptyset$\\
\hline
\end{tabular}
\end{center}

This view f induces an equivalence relation on DB with two classes [ab, a] and [b, $\emptyset$].  0 $\leq$ f because there is more than one equivalence class in DB/f, and   $\leq$ 1 because there are fewer than four equivalence classes in DB/f.

Bancilhon and Spyratos define a complement as another view c such that f $\times$ c is equivalent to 1, which means that the induced equivalence relation for f $\times$ c has exactly one state in each class.  Intuitively, f $\times$ c contains all the information in the original database DB.

In our example, let c select rows with the value b.  c is defined by this table: 
\begin{center}
\begin{tabular}{ |c|c| } 
\hline
Base State & View State\\ 
\hline
\hline
a,b & b\\
a & $\emptyset$\\
b & b \\
$\emptyset$ & $\emptyset$\\
\hline
\end{tabular}
\end{center}

Combining f and c gives this table:
\begin{center}
\begin{tabular}{ |c|c|c| } 
\hline
Base State & View State & Complement State\\ 
\hline
\hline
a,b & a & b\\
a & a & $\emptyset$\\
b & $\emptyset$ & b \\
$\emptyset$ & $\emptyset$ & $\emptyset$\\
\hline
\end{tabular}
\end{center}

The equivalence relation induced by f $\times$ c has four classes, and each state in DB is in a unique equivalence class.  In other words, f $\times$ c is equivalent to 1, to all the information in the database.  In this case, the base state equals the union of the view and complement states.  

Bancilhon and Spyratos then show that, for every complement of a view f, there is an update translation that is constant on the complement, not changing it.  Their result was very exciting, because it suggested that all that was needed was to define the right complement and update translations would pop right out.  Unfortunately it wasn't as simple as that.  First, there are lots of complements; 1 is a complement to every view, because it already contains all the information in the database.

A reasonable hope is that there is some kind of minimal complement.  Bancilhon and Spyratos defined a minimal complement as a view g such that g is a complement of f (f $\times$ c $\equiv$ 1) and for all other complements h of f, g $\leq$ h.  They then show that  minimal complements exist only in trivial cases, when f=0 or f=1.   They provide a construction; given a candidate for a minimal complement, they construct a another complement that is incomparable, showing that the candidate isn't minimal.  Applying their construction to our example, let h be defined by this table:

\begin{center}
\begin{tabular}{ |c|c| } 
\hline
Base State & View State\\ 
\hline
\hline
ab & b\\
a & ab \\
b & ab \\
$\emptyset$ & $\emptyset$\\
\hline
\end{tabular}
\end{center}

We can see that h is a complement by examining this table:
\begin{center}
\begin{tabular}{ |c|c|c| } 
\hline
Base State & View State & Complement State\\ 
\hline
\hline
a,b & a & b\\
a & a & ab\\
b & $\emptyset$ & ab \\
$\emptyset$ & $\emptyset$ & $\emptyset$\\
\hline
\end{tabular}
\end{center}

The induced equivalence classes of f $\times$ h are [ab], [a], [b], [$\emptyset$] which is equivalent to 1, so h is a complement of f.  However, g and h are incomparable (g$\nleq$h and h$\nleq$g).  The induced equivalence classes for g are [ab, b] and [a, $\emptyset$]; neither set of equivalence classes is a refinement of the other, so the views are incomparable.

Despite this theoretical limitation, other researchers were eager to apply their approach.  The other researchers reasoned that they needed to find intuitive complements.  But Bancilhon and Spyratos' approach was difficult to work with, blisteringly abstract, without the the usual relational database apparatus of domains, table structures, column names...  Even worse, it depended on updates that were completely different from SQL's INSERT, DELETE  and UPDATE operations.  Bancilhon and Spyratos require  "complete sets of updates," which were a set of updates (f:S$\rightarrow$S) that were closed under composition and had local inverse (give a state s $\in$ S and f:S$\rightarrow$S there had to exist a g:S$\rightarrow$S such that g(f(s))=s.)  It wasn't clear to researchers what a complete set of updates entailed.  And nothing looked like the Codd's relational operators.

The first sign of trouble was a note by Keller \cite{16} claiming that Bancilhon and Spyratos were wrong.  It turned out to be a confusion of terms and assumptions.  Hegner\cite{11}  described Keller's approach as the "open view" approach, that a user would see both the view and base databases, and Bancilhon and Spytratos's approach as the "closed view" approach, that a user would see only the view, and would not be able to see the base database.  Keller's approach was basically a programmer's approach, of writing a view update while knowing the structure of the base database; he eventually \cite{17}\cite{18} proposed a utility to help users write their own updates.  Date \cite{29} completed this approach by showing view update logic for several relational operators, including join.

The closed approach sees a view update facility as part of a DBMS, that has access the base tables, but hides the base tables from users.  The closed approach needs a systematic solution that applies to every possible view.  It demands math of the kind that Bancilhon and Spyratos used.

Problems multiplied.  Cosmadakis and Papadimitriou \cite{6} looked at a single table with a projection view.  They additionally required a view complement to have a join dependency with the view.  For example, let a table T have four columns A, B, C and D, and let T =  $\pi_{AB}$(T) $\times$ $\pi_{CD}$(T), with one tuple removed.  T doesn't satisfy a join dependency but $\pi_{CD}$(T) is a complement.

Bancilhon and Spyratos don't require that, given a view and complement, there is some computation that gives the base table, as a join dependence would require.  Instead, an implication of their approach is that a complement has to distinguish view states; if several base states map to the same view state, then the complement has to map each of those base states to different complement states.

Using their flawed definition of a view complement, Cosmadakis and Papadimitriou derive some very odd results.  They show that, in order to translate an insertion, the projected rows have to functionally determine the unprojected rows.

As a counterexample, consider an employee table keyed on the Social Security number (SSN).  A view projects all the columns except the SSN, to protect the privacy of the SSN.  To Cosmadakis and Papadimitriou, the employee data doesn't  functionally determine the SSN, so no new employees can be inserted in the view.  To me, the obvious solution is for a view insert to insert a row with the SSN set to a null value; I use this solution below.  Hegner \cite{25} immediately showed that null values solved the projection problem, but his result wasn't widely noted.

The paper that killed the constant complement approach was by Lagerak \cite{19}.  His key theorem is Lemma 3.2 which says that given a view f, an update translator T$_u$, and two different base states s$_1$ and s$_2$ with f(s$_1$)=f(s$_2$) and T$_u$(s$_1$)=T$_u$(s$_2$), then T$_u$ isn't a valid translator.  His lemma is correct but his interpretation is wrong, because the fact that T$_u$(s$_1$) won't equal T$_u$(s$_2$) characterizes constant complement translators.  Intuitively, s$_1$ and s$_2$ are different because their complements are different; translators preserve those different complement so the updated base states T$_u$(s$_1$) and T$_u$(s$_2$) will also be different.  Lagerak's condition that T$_u$(s$_1$)=T$_u$(s$_2$) will simply never occur with constant complement translations.  His lemma is true but vacuous.

Researchers at that time cared mainly about projections and joins, because they were using dependencies to structure tables.  An approach that didn't work with projections seemed hopeless.  Many researchers gave up and little progress was made  \cite{22} \cite{23} \cite{24} \cite{13}.

There were two notable exceptions.  Hegner, noticing that all of the relational operators except set difference and division were monotonic, developed an order-based approach \cite{26} with a lattice of views, in which a view and complement combine in a join.  The approach can handle only a series of insertions or deletions but not both, so it can't handle a simple update operation.

Johnson and Rosebrugh noticed the flavor of category theory in \cite{5}  and recast the theory in their Entity-Attribute sketches, which are made out of categories.\cite{28}\cite{29} \cite{30} Unfortunately, they weren't able to make any progress on the problem and turned to update translations in lenses \cite{31}.

Prior researchers focused on the theorem in \cite{5} that, given a complement, a unique complement-preserving update translation exists.  Another theorem in \cite{5} shows that, given a update translation, there was a unique complement that it preserved.  I propose using both theorems as the up and down strokes of a saw.

\section{From Update Translations to Complements}

Let's start with a few definitions from \cite{5}.

\begin{definition}[Update]
Let S be a set of states.  An update is a mapping u:S$\to$S.
\end{definition}

\begin{definition}[View]
Let S, V  be sets of states.  An view f is an onto mapping f:S$\to$V.
\end{definition}

\begin{definition}[Translation]
Let S, V  be sets of states, and f:S$\to$V be a view.  Let v:V$\to$V be a view update and T(u)S:$\to$S be a base update.  Then T(u) is a translation iff
\begin{enumerate}
\item fT(u)=uf
\item $\forall s \in S$, uf(s)=f(s) $\rightarrow T_u(s)=s$
\end{enumerate}
\end{definition}

The first condition is usually shown as a commutative diagram, in which every path is equal.
\begin{center}
\begin{tikzcd}
A \arrow [r, "T(u)"] \arrow [d, "f"] & B \arrow [d, "f"] 
\\ V(A) \arrow [r, "u"] & V(B)
\end{tikzcd}
\end{center}

\begin{definition}[Complete Set of updates]
A set of updates U is complete when
\begin{enumerate}
\item The composition of two updates in U is also in U. ($\forall$f, g $\in$ U, gf $\in$ U)
\item All updates have local inverses ($\forall$f$\in$U, $\forall$s$\in$ S, $\exists$ g$\in$U such that gf(s)=s.)
\end{enumerate}
\end{definition}
If identities are included in the set of updates, then a complete set of updates makes a set of states into a category .

\begin{definition}[Translator]
Let S,V be sets of states.  Let f:S$\to$V be a view.  Let $U_s$ be a set of updates on S and U be a complete set of updates on V.  A mapping T: U$\to U_s$ is a translator iff
\begin{enumerate}
\item $\forall u\in U$, T(u) is a translation
\item $\forall u,v\in U$, T(uv)=T(u)T(v)
\end{enumerate}
\end{definition}
This makes a translator a functor in category theory.

\begin{definition}[Equivalence]
Let U be a complete set of view updates.  Let T be a translator.  For s, s'$\in$S, s$\equiv$s' iff $\exists u \in$ U such that s=T(u)(s')
\end{definition}
In other words, s$\equiv$s' if a translated update maps base state s' to s.  \cite{5} shows that $\equiv$ is an equivalence relation.  This  means it generates a partition of S, denoted S/$\equiv$. \cite{5} then shows that c:S$\to S/\equiv$ is the complement of the translation T.

This is very elegant, but the states of the complement are sets of base states rather than base states themselves.  It's not clear at all how this complement makes sense in a relational database, which may be why other researchers didn't pay attention to this result.

I will show that, for many relational operators, the equivalence classes of the complement are in one-to-one correspondence with an easily-computed set of view states, showing how the complement makes sense relationally.

\section{Schemas}
Before showing these things, I first need to introduce my formal definition of database.  For simplicity and without loss of generality, I won't use column names but only column numbers.  To avoid an overwhelming number of numeric subscripts and superscripts, I will use C-like notation, 

\label{sec:Schemas}

\begin{definition}[Domain]
A domain is a set, with one element defined as the “null” value.
\end{definition}
Each domain is denoted with an italicized label, such as \textit{String}. In Oracle, the null value is a missing value, meaning that value isn’t known; calculations with null values result in null values.   In other DBMS, each domain may have a unique null value, such as 0 for number and a zero-length string for strings.
\\Most database systems provide a small number of domains:
\begin{enumerate}
\item Numbers
\item Character strings
\item Dates and date/times
\item Logical (true/false) values
\item Undifferentiated data (blobs)
\end{enumerate}

\begin{definition}[Table Schema]
A table definition is an array of Domains.
\end{definition}

\begin{definition}[Database Schema]
A database schema is an array of table schemas.  Note that each element of the array is itself an array.  The lengths of the table definition array can be different.
\end{definition}
An example is a database with two tables; the first table has two columns, both strings; the second table has three columns, a string and two numbers.  Its schema is:
\\Schema[1][1]=\textit{String};
\\Schema[1][2]=\textit{String};
\\Schema[2][1]=\textit{String};
\\Schema[2][2]=\textit{Number};
\\Schema[2][4]=\textit{Number};

\begin{definition}[Table State]
A table state is a subset of Domain[1] $\times$ Domain[2] $\times$ \textellipsis $\times$ Domain[n], where n is the table width.
\end{definition}
A table state is a relation, a set of n-tuples, where n is the length of the table.  Each element of each tuple belongs to the domain specified in the schema.   For example, a valid state for the first table is a set of 2-tuples, where the both elements of the tuples are string, such as:
\\ \{(‘A’, ‘AAA’), (‘B’, ‘BBB’)\}.

\begin{definition}[Database State]
A database state is an array of table states, where each table state is valid for the schema.  I will use a C-like notation of T[1] and T[2].
\end{definition}
For example, a possible state for this schema is 
\\T[1]= \{(‘A’, ‘AAA’), (‘B’, ‘BBB’)\} 
\\T[2]= \{(‘A’, 1, 11), (‘A’, 2, 22), (‘A’, 3, 33), (‘B’, 4, 44)\}.

\begin{definition}[Database Update]
Given two states A and B of a schema, the database update between them is two arrays of sets of tuples, Add and Del, consistent with the schema of the states.  Add[i] contains the tuples to add to table i and Del[i] contains tuples to delete from table i.  This means Add[i] $\cap$ T[i] = $\emptyset$ and Del[i] $\subseteq$ T[i].
\end{definition}

A SQL UPDATE operation deletes one row in a table and adds a new, altered row into the same table, so appears as rows in both the Add and Del sets.

Using the example, an update u from state a to state ab adds the row b.  So u.Add[1]={b} and u.Del[1]=$\emptyset$.

\begin{lemma}[Updates can be calculated]
Given states A and B, the database update u:A → B can be calculated.  For each table i:
$Add[i] = A.T[i] \setminus B.T[i]$.  (The rows in B that aren’t in A)
$Del[i] = B.T[i] \setminus A.T[i]$.  (The rows in A that aren’t in B)
\end{lemma}

This implies that updates are unique, because the calculation has only one result.  Between any two states, there is is exactly one update.

Updates can be composed.  If u:A$\to$B and v:B$\to$C, then vu:A$\to$C must be the unique update from A to C.  This leads immediately to:

\begin{lemma}[Database updates are a complete set of updates]
\end{lemma}
\begin{proof}
\begin{enumerate}
\item For updates u:A$\to$B and v:B$\to$C, $\exists vu:A\to$C in the set of database updates. 
\item  For u:A$\to$B, $\exists v:B\to A$ in the set of database updates.  vu(s)=s.
For all tables i, v.Add[i]=u.Del[i] and v.Del[i]=u.Add[i].
\end{enumerate}
\end{proof}

Compositions are computed as:
\begin{enumerate}
\item vu.Add[i] = (v.Add[i] $\cup$ u.Add[i]) $\setminus$ v.Del[i]
\item vu.Del[i] = (v.Del[i] $\cup$ u.Del[i]) $\setminus$ v.Add[i]
\end{enumerate}

Inverses are computed as:
\begin{enumerate}
\item v.Add[i] = u.Del[i]
\item v.Del[i] = u.Add[i]
\end{enumerate}

Every state has an identity. id:s$\to$s  id.Add[i]=$\emptyset$;id.Del[i]=$\emptyset$ 

Our example above can be diagrammed as:

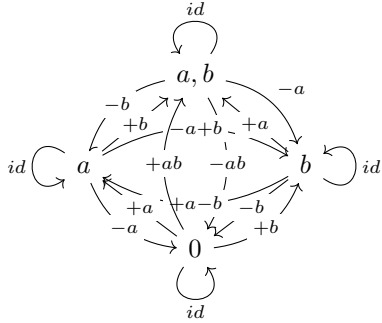
\begin{center}
\begin{tikzcd}
                                                                                                                                                                     & {a, b} \arrow[ld, "-b" description, harpoon, bend right] \arrow[rd, "-a", bend left] \arrow["id "', loop, distance=2em, in=125, out=55] \arrow[dd, "-ab" description, bend left] &                                                                                                                                                         \\
a \arrow[ru, "+b" description] \arrow[rd, "-a" description, bend right] \arrow["id"', loop, distance=2em, in=215, out=145] \arrow[rr, "-a+b" description, bend left] &                                                                                                                                                                                  & b \arrow[lu, "+a" description] \arrow[ld, "-b" description] \arrow["id"', loop, distance=2em, in=35, out=325] \arrow[ll, "+a-b" description, bend left] \\
                                                                                                                                                                     & 0 \arrow[lu, "+a" description] \arrow[ru, "+b" description, bend right] \arrow["id"', loop, distance=2em, in=305, out=235] \arrow[uu, "+ab" description, bend left]              &                                                                                                                                                        

\end{tikzcd}
\end{center}

\section{Selection}
The simplest example of this construction is the Selection relational operator, which selects the rows of a table that satisfy a logical condition $\alpha$.  The rows that satisfy $\alpha$ are in the view.

It's immediately obvious that its complement is another selection operator, that selects rows that don't satisfy $\alpha$.  So the a complement view c is a selection whose logical condition is $\lnot \alpha$.  The selection view v and its complement c together contains all the rows in the base table.  (For state s, v(s) $\cup$ c(s) = s)

It's somewhat less obvious that all view updates contain only rows that satisfy $\alpha$.  Look at additions and deletions separately.  A row not satisfying $\alpha$ can't be added, because then the updated view state would contain rows that don't satisfy $\alpha$, so wouldn't be in the Selection view.  All deleted rows satisfy $\alpha$ because they were in the view before deletion.

The same is true of the complement: in an update between complement states, every added or deleted row satisfies $\lnot \alpha$.

This immediately shows how to translate a selection view.  A view update has added and  deleted rows which satisfy $\alpha$, which can also be added and deleted in the base scheme, so the view update and base update are identical.  This translation preserves the complement, because it doesn't affect any rows that aren't $\alpha$-compliant, which are the rows in the complement.  So the update translation is identity: the view update is also the base update.

This implies that a DBMS has to check that added view rows satisfy $\alpha$.  If a new view row didn't satisfy $\alpha$, it wouldn't be a view update because it wouldn't map to any other view states.  

Next I prove that this is the complement promised by \cite{5}.  Translated view updates only add and delete  $\alpha$-compliant rows; the non-compliant rows aren't touched.  For base state s, every other state s' in the equivalence class (s$\equiv$s') has a translated view update T(v) such that T(v)(s')=s.  But T(v) only has rows that satisfy  $\alpha$, and will never add or delete non-compliant rows.  So every state in the  equivalence class S/$\equiv$ will have the same non-compliant rows, because no translated updates will affect them.  The non-compliant rows are exactly the complement c(s).  So there is a one-to-one correspondence between equivalence classes and states of the complement.  In other words, the method in \cite{5} computes exactly what we want, once we find the common element of the states of the equivalence class.

These observations immediately generalize to the case of multiple tables.  Each table i has a predicate $\alpha_i$, and has a complement of $\lnot\alpha_i$ rows.  The complement state is the array of complement tables.

A very nice property of Selection is that the view and complement have no overlap; they split the information in a database in half.  This leads to a different definition of an ideal complement:

\begin{definition}[Perfect Decomposition]
In a perfect decomposition, the intersection of the view and its complement is empty.
\end{definition}

A good heuristic for guessing the complement of an update translation is to determine the base database state after a view update that deletes all view rows.  For example, consider a view state v, and a view update u deletes all the rows in v.  T(u) will delete all the $\alpha$-compliant rows in the base table, because those are the rows in the view, leaving all the  $\lnot \alpha$-compliant rows -- which is the complement.

\section{Union}
The inverse of the Selection relational operator -- union -- is also very simple.

It's immediately obvious that a union has an empty complement.  Every row that is in the two unioned tables are in the view, so there are no rows that are excluded.

This implies that union is a perfect decomposition, because the intersection of any set and the empty set is the empty set.

Unfortunately, this doesn't help us select a transaction for a union operation.  Deletions are obvious -- a deletion in the union has to translate into deletions in both base tables; otherwise the undeleted row would still appear in the Union view.

The problem is adding a row to the Union, because many reasonable translations exists.  A row added to the Unions view can be added to the left table, or to the right table, or to both tables, or randomly to the left or right tables.  Still other translations are possible.  \cite{5} provides no guidance because all translations have the same empty complement.

\section{Projection}

Let's look next at the issues that the contemporaneous researchers saw, which were with projection.  Let T be a table with three columns – K, A and B – with K functionally determining A and B.  Let V be the projection of T onto columns K and A, denoted $\pi_{KA}$.  The obvious complement is a projection onto K and B, denoted $\pi_{KB}$, with $\pi_{KA} \otimes \pi_{KB} = T.$  As  \cite{7} noted, if a row is deleted from $\pi_{KA}$, the intuitive translation is to delete the row with the same key from T; the reasoning is that if the row’s key is deleted, then all dependant attributes should be deleted too.  But deleting a row from base table T would also delete the row from the complement $\pi_{KB}$, so the intuitive translation creates side effects and isn’t valid under the constant complement approach. Adding a row to $\pi_{KA}$ has the same problem.  Modifying non-key data doesn’t, as \cite{23} points out.

My solution is the SQL INSERT statement as a model; if a column isn't specified in the INSERT, then it get a null value.  So when adding a row to $\pi_{KA}$, I specify that the translated update sets column B of the base row to null.  To insure that the complement $\pi_{KB}$ isn't modified, I change the definition of the complement to not only project rows K and B, but also to select rows with B$\neq$null.  Since a complement is a view, making this additional specification is perfectly legitimate.

I use a similar trick in translating a deletion from the $\pi_{KA}$ view.  I don't delete the base row, but set the A column is set to null.  To make this row disappear from the view, I change to view to also exclude rows with null values in A.

The last step is to consider what happens when if the row is deleted from both $\pi_{KA}$ and $\pi_{KB}$; the base row should be deleted because it holds no useful information.

Given these revised view definitions and update translations, we can easily compute the induced complement.  Pick a base state s.  Its equivalence class members are all the states that can be reached by translated updates.  The translated updates always change columns K and A, giving them every possible value.  The only part of the state s that doesn’t change are its (K, B) values, because the translated updates don’t change them.  So the set of equivalence classes is in one-to-one correspondence with $\pi_{KB}$.

\section{Joins}
Joins are much more complicated to update than projections.  One reason is that joins are used for computations in databases, not just for representing information.  Consider this pathological example,

A database has two tables.  The Item table has five columns: ItemID, ItemDesc, Height, Length, Depth.  The Box table has five columns: BoxID, Height, Length, Depth, Cost.  The following SQL statement finds all the boxes into which an item can fit and orders them by box cost:
\begin{lstlisting}[language=SQL]
SELECT BoxID 
FROM Box JOIN Item ON 
    Item.Height<=Box.Height 
and Item.Length<=Box.Length 
and Item.Depth<=Box.Depth
WHERE Item.ItemID='XYZ'
ORDER BY Item.ItemID, Box.Cost
\end{lstlisting}

This SQL statement clearly represents a relationship between BOX and ITEM, that we could call a "Best Fit" relationship. What would it mean to add a row to this join?  It's hard to say.  Should you you add a box or an item, and what should its dimensions be?

It clearly it makes no sense to translate view updates for Joins that only compute a result.  Only the "structural" joins should be considered.  Most database design methods break tables into smaller tables, then  use joins to bring separated information back together; these are the structural joins.  I'll look at two kinds of structural joins, while making no claim that they are the only structural joins.

\subsection{Hierarchical Joins}
A hierarchical join Parent $\otimes$ Child is a one-to-many relationship, where every Child row must join with at least one Parent rows.  The Parent typically has a unique key, which the Child has as well; they are equijoined on that key.  In addition, the parent’s key is unique.  The child’s key can either be non-unique or can be part of larger unique key that includes the parent's unique key.  An integrity constraint is that every child row must have a parent row, typically implemented with a SQL Foreign Key constraint.

A good example of a hierarchical join is an invoice, where each row in Parent represents an invoice and each row in Child represents a line on the invoice.  The key is the invoice number.  Every line row has to have a invoice row.  

The complement of a hierarchical join is easy to see.  The only information that doesn't appear in a equijoin are the parent rows for which there are no children.  Given the join view and its complement, it's possible to recreate the base parent table by projecting the parent columns and unioning it with complement, and the recreate the base child table by projecting on the child columns.  This is a perfect decomposition.

Translating view updates is also straightforward.  I'll look at deletion first.  A row in the view represents one parent and one child row, so deleting it should be handled by deleting the child row; the parent may be the parent of another child too, and deleting could violate the integrity constraint.  But if all of a parent's children are in the view deletion set, then the parent row should be deleted too.

Adding view rows is straightforward.  A new view row should be broken down into a parent part and a child part.  If the parent part already exists in the base parent table, then add the child row to the base child table; if the parent row doesn't exist base parent table, then add rows to parent and child base tables.

The induced complement of this view update strategy is, as expected, the childless parent rows of s.  To see this, let s be a base state.  Consider a view update that deletes all rows in the view.  This update strategy will delete all the base children rows and all the parent rows that have children; the only base rows remaining will be the childless parent row.  So translated view updates can't delete any childless parent rows.  At the same time, no matter what view rows are added to the state, no childless parent rows will be added to the base.  So the common element shared by all members of [s] is the set of childless parent rows.  The set of equivalence classes is in one-to-one correspondence with the states of the complement.

\subsection{Foreign Key Joins}
Another kind of structural join is a foreign key join, where a column in a table is a key to another (foreign) table.  For example, an invoice line may contain a part number, which is a key to the Part table.  

A foreign key join is also a one-to-many joins but its different intention leads to a different update strategy, which leads to a different complement.  Updating the Join view should never modify the foreign table.  For example, deleting an invoice line should never affect the Part table, which represents a separate entity.

More precisely, a view deletion should delete the local row, but never the foreign row.  A view addition should add a new local row, but never add new foreign rows.  Essentially, the foreign table is off-limits to the view and should never be altered.

The complement is clearly the foreign table, because view updates never alter it.  To see this, let s be a base state.  Consider a view update u that deletes all rows in v(s).  The translation of u deletes all the local rows but none of the foreign rows.  All of the translated view updates won't changed the foreign table in s, which means that all the states in [s] will have the same foreign table rows.  Hence the foreign table is in one-to-one correspondence with the induced equivalence class, and so is the complement.

This complement is also a perfect decomposition.

\subsection{Other Strategies for One-to-Many Joins}
An obvious update strategy for one-to-many joins combines the prior two update strategies.  It doesn't allow allow foreign table rows to be deleted, but does allow them to be added.  A view addition could translate to a foreign table addition.

This update strategy does not qualify as a translator.  Recall that a translator T has to preserve composition:
    $\forall u,v\in U$, T(uv)=T(u)T(v)
    
Consider an update u that adds a single view row r, and that translates into an update that adds a row to both the local and foreign tables.  Consider $u^{-1}$, its inverse, that deletes view row r.  The update strategy translates it into a deletion in the local table, but no action in the foreign table.  In this case T(id)=T($u^{-1}$u)<>T($u^{-1}$)T(u).

Though this update strategy seems perfectly valid, it isn't a Bancilhon-Spyratos translator; the theory is silent about it, and I can't compute its complement.

\subsection{Recognizing}

I've shown two one-to-many joins that have different intentions, different update strategies and different complements.  These are the sorts of semantic information that Bancilhon and Spyratos thought would help users select the right complements.

DBMSs can already distinguish between hierarchical and foreign key joins.  In a hierarchical join, the parent and child have the same key, and the join uses that key.  In a foreign key join, the join terms match a non-key column in the local table to a unique key in the foreign table.  In addition, the non-key column in the local table will often have been declared a foreign key for referential integrity checking.  This suggests that the elusive "semantic knowledge" may already be encoded in databases, and that DBMSs could use this encoding to automatically provide view updates.

\section{Summary}
Bancilhon and Spyratos’s work from 1982 is a good foundation for solving the view update problem.  The approach is quite simple for some relational operators such as Selection, Union and some Joins; a useful feature is that the solvable joins can be determined by examining the Join keys. Even Projection has a "not too bad" solution using null values.  For these reasons, I suggest that the view update problem be resuscitated, and that the Bancilhon-Spyratos approach be used the basis of new work.

\section*{Acknowledgement}
I would like to thank Jennie Rogers of Northwestern University for generous and kind advice about how the academic world works.

\end{document}